\documentclass[letterpaper, 10 pt, conference]{ieeeconf}  

\IEEEoverridecommandlockouts                              

\usepackage[utf8]{inputenc}
\usepackage{amsmath}
\usepackage{amssymb}

\usepackage{amsthm}
\usepackage{subcaption}
\usepackage{amsfonts}
\usepackage{mathtools}
\usepackage[hyphens]{url}
\usepackage{hyperref}
\usepackage{bm}
\usepackage{enumerate}
\usepackage{xcolor}
\usepackage{graphicx}
\usepackage{subcaption}
\usepackage{siunitx}
\usepackage[font=small]{caption}

\newtheorem{dfn}{Definition}
\newtheorem{thm}{Theorem}
\newtheorem{cor}{Corollary}

\newcommand{\cset}[2]{\left\{ #1 \; \middle| \; #2 \right\}}
\newcommand{\ddt}{\frac{\mathrm{d}}{\mathrm{d}t}}
\newcommand{\ddz}[1]{\left. \frac{\mathrm{d}}{\mathrm{d} #1} \right|_{#1 = 0}}

\newcommand{\nR}[1]{\mathbb{R}^{#1}}		
\newcommand{\upperRomannumeral}[1]{\uppercase\expandafter{\romannumeral#1}}	

\newcommand{\image}[1]{\text{Im}\left( #1 \right)}

\newtheorem{problem}{Problem}

\newtheorem{lemma}{Lemma}

\title{\LARGE \bf
The Geometry of Coordinated Trajectories for Non-stop Flying Carriers Holding a Cable-Suspended Load
}

\author{Pieter van Goor$^{1}$, Chiara Gabellieri$^{2}$ and Antonio Franchi$^{2,3}$
\thanks{$^1$ School of Aerospace, Mechanical, and Mechatronic Engineering (AMME), Faculty of Engineering, University of Sydney, NSW, 2006, Australia. {\tt\footnotesize pieter.vangoor@sydney.edu.au}}
\thanks{$^2$ Robotics and Mechatronics Department, Electrical Engineering,  Mathematics, and Computer Science (EEMCS) Faculty, University of Twente, 7500 AE Enschede, The Netherlands. {\tt\footnotesize c.gabellieri@utwente.nl}, {\tt\footnotesize schol@r-franchi.eu}}
\thanks{$^3$ Department of Computer, Control and Management Engineering, Sapienza University of Rome, 00185 Rome, Italy, {\tt\footnotesize schol@r-franchi.eu}}
}

\begin{document}

\maketitle
\thispagestyle{empty}
\pagestyle{empty}

\begin{abstract}
This work considers the problem of using multiple aerial carriers to hold a cable-suspended load while remaining in periodic motion at all times.
Using a novel differential geometric perspective, it is shown that the problem may be recast as that of finding an immersion of the unit circle into the smooth manifold of admissible configurations.
Additionally, this manifold is shown to be path connected under a mild assumption on the attachment points of the carriers to the load.
Based on these ideas, a family of simple linear solutions to the original problems is presented that overcomes the constraints of alternative solutions previously proposed in the literature.
Simulation results demonstrate the flexibility of the theory in identifying suitable solutions.
\end{abstract}

\section{Introduction}
Aerial robotic manipulation has identified cables as privileged tools for their light weight and affordability.
To overcome the robustness and payload limitations of a single-robot approach and to enable the full-pose control of the object, cooperative systems have been proposed, where multiple flying carriers collaborate to manipulate a cable-suspended rigid body. 
Two carriers are used in~\cite{pereira2018asymmetric, Chen2019, gabellieri2023equilibria} to manipulate slender objects, while \cite{ Masone2016, Sanalitro2020,li2021cooperative} use multiple carriers to sustain rigid-body loads. 
Most of the literature on cable-suspended load manipulation employs \textit{multirotors} as flying carriers and, consequently, the proposed manipulation strategies rely on the assumption that the carriers can hover in mid-air.
However, multirotors are constrained by limited endurance \cite{leutenegger2016flying} that hinders their applicability to long-distance transport.
In contrast, fixed-wing Uncrewed Aerial Vehicles (UAVs) offer superior aerodynamic efficiency and endurance \cite{leutenegger2016flying} but are subject to the critical constraint of maintaining a non-zero forward speed.
Such a constraint makes the manipulation especially challenging when the task imposes a constant pose of the load, while the carriers must remain in motion at all times.

Motivated by the potential and efficiency of fixed-wing UAVs, recent work has considered enabling cooperative manipulation of cable-suspended load using \textit{non-stop flying carriers}; \cite{williams2009dynamics} theorizes the transport of a \textit{point-mass} load tethered to two \textit{non-stoppping carriers}, and considers the system kinematics to plan the two carriers' trajectories; \cite{quenneville2023experimental} experimentally shows load lifting for the same system.
Manipulation of a generic rigid-body load with non-stop flying carriers has been addressed firstly in \cite{gabellieri2024existence}. 
The method builds upon the load internal-force redundancy: it ensures that the cables exert the required total wrench to maintain the object pose, while acting on the internal forces to guarantee non-stop carrier motion without perturbing the total wrench on the object. 
Especially, in \cite{gabellieri2024existence}, the authors show that three is the minimum number of non-stop flying carriers to maintain the constant pose of a rigid body load. 
The work also provides a method to generate the sought coordinated orbits of the carriers. 
In \cite{gabellieri2025coordinated}, the results have been extended to a general number of carriers greater than or equal to three. 
However, the proposed method limits the possible carrier trajectories by constraining the internal force component of each cable on a 2D plane. 
This simplifies the trajectory generation method and allows for direct extension of the three-carrier method proposed in \cite{gabellieri2024existence}.
However, it also constrains the carrier trajectories on a subset of all the possible ones, potentially failing to exploit the full potential of the manipulation system.

In this work, we study the problem of using multiple carriers to apply a given wrench (force and torque) to a rigid body load through cables at fixed attachment points, while maintaining nonzero velocity along smooth periodic trajectories.
We propose a new differential geometric perspective on the problem and show that the problem of finding suitable trajectories is naturally recast to that of finding smooth immersions between manifolds.
In contrast to previous work, our analysis does not impose a particular type of solution, and the full potential freedom of the manipulation system can be exploited.
We show that the manifold of admissible configurations is connected, meaning that there exist trajectories connecting any two desired configurations.
We also show that any individual carrier is able to instantly move in any given direction while the other carriers compensate, subject to a mild assumption on the attachment point geometry.
Finally, based on the preceding theory, we propose a simple `linear' method for generating admissible periodic trajectories, and we provide simulation results demonstrating how it can create trajectories impossible through prior approaches.

\subsection{Notation}

The norm of a given vector $x \in \mathbb{R}^n$ is denoted $\Vert x \Vert$.
The unit sphere is defined as 
\begin{align*}
    S^2 = \{ x \in \mathbb{R}^3 \mid \Vert x \Vert = 1 \} \subset \mathbb{R}^3.
\end{align*}
The projection $\pi : \mathbb{R}^3 \setminus \{ 0 \} \to S^2$ of any nonzero vector to the sphere is defined by
\begin{align}
    \pi(x) = \frac{x}{\Vert x \Vert}.
\end{align}
For any smooth map between smooth manifolds $h : \mathcal{M} \to \mathcal{N}$, the differential at $x \in \mathcal{M}$ is the linear map between tangent spaces, denoted $\mathrm{D} h (x) : \mathrm{T}_x \mathcal{M} \to \mathrm{T}_{h(x)} \mathcal{N}$ such that for any
direction $\dot{x} \in \mathrm{T}_x \mathcal{M}$, the derivative of $h$ at $x$ in the direction $\dot{x}$ is given by
\begin{align*}
    \mathrm{D} h (x) [\dot{x}] = \ddz{t} h(x(t)),
\end{align*}
where $x(t)$ is any differentiable curve in $\mathcal{M}$ such that $x(0) = x$ and $\dot{x}(0) = \dot{x}$.
A smooth map between manifolds $h : \mathcal{M} \to \mathcal{N}$ is called an \emph{immersion} if its differential $\mathrm{D} h(x)$ is injective for every $x \in \mathcal{M}$.

\section{Problem Description}

Consider a manipulated \emph{load}, modelled as a rigid body, attached by cables to $n \geq 3$ \emph{carriers}.
Define the load body frame $\{ O \}$ to coincide with the centre of mass of the load.
For each cable $i=1,\ldots,n$, the position of its \emph{attachment point} on the load, expressed in the frame $\{ O \}$, is given by a constant $b_i \in \nR{3}$.
We assume that each cable has negligible mass and is taut (has positive tension).
Then the position $p_i\in \mathbb{R}^3$ of each carrier, expressed in $\{O\}$, is determined by the direction of the cable and its length, that is, \begin{equation}p_i = b_i + q_i l_i,\label{eq:kinematics}\end{equation} where $l_i > 0$ is the length of the cable and $q_i\in S^2$ represents the direction of the cable expressed in $\{O\}$.
By assumption, each cable has a positive tension $T_i > 0$, and thus the force exerted by the cable $i$ at the corresponding attachment point is $f_i = T_i q_i \in \nR{3}$, also expressed with respect to the load frame $\{ O \}$.

The total wrench (combined force and torque) experienced by the load due to the cables is given by
\begin{align}\label{eq:wrench_equation}
    w &= G f, \notag \\
    \begin{pmatrix}
        f_{\text{CoM}} \\
        \tau_{\text{CoM}}
    \end{pmatrix}
    &= \begin{pmatrix}
        I_3 & \cdots & I_3 \\
        b_1^\times & \cdots & b_n^\times
    \end{pmatrix}
    \begin{pmatrix}
        f_1 \\ \vdots \\ f_n
    \end{pmatrix}.
\end{align}

The problem we seek to address is that of generating a constant wrench such that the load remains static, while each carrier continues moving (non-stopping carriers problem, introduced in~\cite{gabellieri2023equilibria}).
In addition to the force in each cable, the load is acted on by gravity, which is expressed in the frame $\{O\}$ by $m g \eta$, where $\eta \in S^2$ is the direction of gravity expressed in $\{O\}$.
If the load is static, then $\eta$ is constant, and thus the wrench $w$ required to cancel the gravity force is constant as well.
Since the position of the carrier $i$ is given by $l_i q_i$, its velocity is given by $\dot{q}_i l_i$.
Formally, the problem considered is as follows. 
\begin{problem}\label{problem}
Given a desired wrench $w \in \nR{6}$, find trajectories for $q_1,...,q_n$ and $T_1,...,T_n$ such that
\begin{enumerate}
    \item The trajectory of each $q_i$ and $T_i$ is smooth and periodic.
    \item The wrench exerted by the carriers \eqref{eq:wrench_equation} is equal to the desired wrench.
    \item The tension $T_i$ in each cable is always positive.
    \item The velocity in each carrier is always non-zero.
\end{enumerate}
\end{problem}

\section{Geometric Description of Admissible Force Configurations}

In this section, we show that the problem of finding suitable periodic trajectories is equivalent to that of finding immersions of the unit circle into a particular smooth manifold.

\subsection{Admissible Forces}

Based on our problem description, we create the following definition for an admissible configuration of forces.

\begin{dfn}
    A configuration of forces $f = (f_1,...,f_n) \in \nR{3n} $ is called \emph{admissible} if
    \begin{enumerate}
        \item The desired wrench is exerted on the load, $G f = w$.
        \item Each individual force is nonzero, $f_i \neq 0$ for each $i$.
    \end{enumerate}
\end{dfn}

Let $G^\dag$ denote the pseudoinverse of $G$, and fix a basis matrix $N \in \nR{3n\times k}$ for the nullspace of $G$, where $k$ is the dimension of the nullspace.
As we show in the following theorem, the set of admissible forces defines a smooth submanifold in $\nR{k}$.
The rank of $G$ depends on the geometry of the attachment points $b_i$, but we will assume it is full-rank, meaning the nullspace dimension is $k = 3n-6$.
Without loss of generality, we will choose $N$ to be an orthogonal matrix, so that $N^\top N = I_k$.

\begin{thm}
Define the manifold 
    \begin{align*}
    \mathcal{M}_\lambda := \cset{\lambda \in \nR{k}}{P_i G^\dag w + P_i N \lambda \neq 0},
\end{align*}
where $P_i \in \nR{3\times 3n}$ is given by
\begin{align*}
    P_i = \begin{pmatrix}
        0_{3\times 3} & \cdots & 0_{3\times 3}
        & \underbrace{I_3}_{\text{$i^\text{th}$ block}} & 
        0_{3\times 3} & \cdots & 0_{3\times 3}
    \end{pmatrix}.
\end{align*}
Then a configuration of forces is admissible if and only if $f = G^\dag w + N \lambda$ for some $\lambda \in \mathcal{M}_\lambda$.
\end{thm}

\begin{proof}
    For the forward direction, let $\lambda \in \mathcal{M}$ and let $f = G^\dag w + N \lambda$.
    Then the wrench is $Gf = G G^\dag w + G N \lambda = w$, and each individual force components is $f_i = P_i f = G^\dag w + N \lambda \neq 0$, so the force configuration is admissible.

    For the reverse direction, consider an admissible force configuration $f = (f_1,...,f_n) \in \nR{3n}$.
    Since $f$ is admissible, $G f = w$ and thus $G(f - G^\dag w) = 0$, meaning that $(f - G^\dag w) \in \image{N}$, and $\lambda$ is well-defined by $N \lambda = f - G^\dag w$.
    Since $f_i = P_i f \neq 0$, it follows that $P_i G^\dag w + P_i N \lambda \neq 0$, and thus $\lambda \in \mathcal{M}_\lambda$.
\end{proof}

The preceding theorem shows that trajectories of admissible force configurations are exactly equivalent to trajectories in the manifold $\mathcal{M}_\lambda$.
This motivates us to study the geometry of the manifold, as it directly relates to the possible solutions to our problem.

\begin{lemma}\label{lem:full-rank-nullspace}
    If $n \geq 4$ and every subset of $n-1$  attachment points on the load is non-collinear, then the matrix $P_i N$ is full-rank for every $i$. 
\end{lemma}

\begin{proof}
For any given fixed $i$, we may relabel the carriers so that $i = n$.
The matrix $P_n N$ is full rank if and only if, for any fixed $f_n \in \nR{3}$, there exists a $\lambda \in \nR{k}$ such that $P_n N \lambda = f_n$.
Equivalently, there must exist $f \in \nR{3n}$ such that $G f = 0$ and $P_n f = f_n$.
Thus, we let $f_n$ be arbitrary and proceed by constructing a suitable $f$.

By assumption, the first $n-1$ carrier attachment points are non-colinear, and $n -1 \geq 3$, so
\begin{align*}
    G_{1:n-1} := \begin{pmatrix}
        I_3 & \cdots & I_3 \\
        b_1^\times & \cdots & b_{n-1}^\times
    \end{pmatrix} \in \nR{3 \times 3(n-1)}
\end{align*}
has full rank.
For $f \in \nR{3n}$, let $f_{1:n-1} \in \nR{3(n-1)}$ denote its first $3(n-1)$ entries.
Then we require
\begin{align*}
    G f &= G_{1:n-1} f_{1:n-1} + \begin{pmatrix}
        I_3 \\ b_n
    \end{pmatrix} f_n = 0,
\end{align*}
where $f_n = P_n f$.
Since $G_{1:n-1}$ is full rank, then this admits a solution for any $f_n$.
Therefore $P_n N$ is full rank, as required.
\end{proof}

The assumption of the preceding lemma is rather mild, and only fails when all but one of the carriers are attached to the load along a straight line.
For a random placement of the attachment points, the assumption holds almost surely.

\begin{cor}\label{cor:path-connected}
    Under the assumption of Lemma \ref{lem:full-rank-nullspace}, the manifold $\mathcal{M}_\lambda$ is path connected.
\end{cor}

\begin{proof}
For each $i$, $P_i N$ is full rank due to Lemma \ref{lem:full-rank-nullspace}, and thus the space of $\lambda$ for which $P_i G^\dag w + P_i N \lambda = 0$ is an affine subspace of dimension $k-3 = 3n-9$ in $\nR{k}$.
Recursively define the manifolds $\mathcal{M}_\lambda^i$ by
\begin{align*}
    \mathcal{M}_\lambda^0 &= \nR{k}, \\ 
    \mathcal{M}^{i+1} &= \mathcal{M}^i \setminus \cset{\lambda \in \mathcal{M}^i}{P_i G^\dag w + P_i N \lambda = 0},
\end{align*}
for $i = 1,...,n$.
At each step, the manifold $\mathcal{M}_\lambda^i$ is of dimension $3n-6$, while the subset that is removed is an embedded submanifold of dimension $3n-9$.
It follows from Lemma \ref{lem:connected_submanifold} that each $\mathcal{M}_\lambda^i$ is path-connected, and thus, in particular, $\mathcal{M}^n = \mathcal{M}_\lambda$ is path-connected.
\end{proof}

Corollary \ref{cor:path-connected} shows that $\mathcal{M}_\lambda$ is path-connected.
This means that, given any two admissible force configurations, it is possible to construct a trajectory between them such that the force configuration remains admissible at all times.
Due to the complex interaction between carrier forces, this result is not trivial without the insight gained from differential geometric modeling.

\subsection{Periodic Trajectories as Immersions}

In this section we show how the stated problem may be reinterpreted as finding a smooth immersion of the unit circle $S^1$ into the manifold $\mathcal{M}_\lambda$.
For any $\lambda \in \mathcal{M}_\lambda$, the bearing of the $i^\text{th}$ carrier is uniquely given by
\begin{align}\label{eq:q_from_lambda}
    q_i(\lambda) = \pi(P_i G^\dag w + P_i N \lambda) \in S^2.
\end{align}
Our stated goal is to identify the smooth periodic trajectories of $T_i$ and $q_i$ such that $\dot{q}_i \neq 0$ and the resulting forces are always admissible.
Since admissibility is guaranteed for all $\lambda \in \mathcal{M}_\lambda$, what remains is to find the smooth periodic trajectories of $\lambda(t) \in \mathcal{M}_\lambda$ such that the $q_i$ obtained from \eqref{eq:q_from_lambda} has nonzero velocity for all time.

The differential of $q_i$, as a function of $\lambda$, is given by
\begin{align*}
    \mathrm{D} q_i(\lambda) 
    &= \mathrm{D}_\lambda \pi(P_i G^\dag w + P_i N \lambda) \\
    &= \mathrm{D} \pi(P_i G^\dag w + P_i N \lambda) P_i N  \\
    &= \frac{(I_3 - \pi(P_i G^\dag w + P_i N \lambda) \pi(P_i G^\dag w + P_i N \lambda)^\top) P_i N}{\vert P_i G^\dag w + P_i N \lambda \vert}  \\
    &= T_i^{-1} (I_3 - q_i q_i^\top) P_i N.
\end{align*}
The matrix $I_3 - q_i q_i^\top$ is the projector onto the 2D subspace orthogonal to $q_i$ and has rank 2.
Thus, if the matrix $P_i N$ has full rank, (e.g. under the assumptions of Lemma \ref{lem:full-rank-nullspace}), then the rank of $\mathrm{D} q_i(\lambda)$ is 2 for all $\lambda \in \mathcal{M}_\lambda$.
This is full rank with respect to the dimension of the tangent space of $q_i$ in $S^2$, which means that, for any given $q_i(\lambda(t))$, the velocity of $\lambda$ can be chosen to achieve any desired velocity of $q_i$.

\begin{thm}
If $\lambda : S^1 \to \mathcal{M}_\lambda$ and $q_i\circ\lambda : S^1 \to S^2$ are smooth immersions (for each $i = 1,...,n$), then the $q_i$ and $T_i$ trajectories corresponding to $\lambda$ solve Problem~\ref{problem}.
\end{thm}

\begin{proof}
Note that $\lambda$ is a map from $S^1$, the unit circle, and hence is necessarily periodic.
Since $\lambda$ is also smooth, the $q_i$ and $T_i$ are also smooth and periodic.
By definition of the force corresponding to $\lambda$, the exerted wrench is always $G f = G(G^\dag w + N \lambda) = w$.
Since $\lambda(t) \in \mathcal{M}_\lambda$, it also holds that $T_i(t) > 0$ for all $t$.
Finally, to see that the velocity in each carrier is always non-zero, simply note that $q_i \circ \lambda$ is an immersion, meaning its derivative is always injective and, in particular, never zero.
\end{proof}

Note that an immersion does not need to be an embedding.
The difference is that the image of an immersion does not need to be a submanifold of its codomain; in other words, the path traced by $q_i \circ \lambda$ may intersect with itself.
However, an immersion is required to be smooth, meaning it may not `stop' or `turn sharply'.

\subsection{A Linear Solution}\label{sec:lin_sol}

This section demonstrates a simple linear algebra approach to finding an explicit design for $\lambda : S^1 \to \mathcal{M}_\lambda$ such that each carrier's motion $\ddt q(\lambda(t))$ is never zero.
We will assume that each $P_i N$ is full rank (see Lemma~\ref{lem:full-rank-nullspace}.
Consider $\lambda$ defined by the formula
\begin{align}\label{eq:linear_solution}
    \lambda(t) = A \mu(t),
    && \mu(t) := \begin{pmatrix}
        \cos(t) \\ \sin(t)
    \end{pmatrix}
\end{align}
where $A \in \nR{3n-6 \times 2}$ is a design parameter.
In order that this induces periodic motions of each of the $q_i$, we require that
\begin{enumerate}
    \item The resulting forces $\lambda(t)$ belong to the manifold $\mathcal{M}_\lambda$.
    \item The carriers are non-stopping, $\ddt q_i(\lambda(t)) \neq 0$.
\end{enumerate}

Beginning with the second condition, we have that $\ddt q_i(\lambda(t)) = 0$ if and only if
\begin{align*}
    T_i^{-1} (I_3 - q_i q_i^\top) P_i N A \dot{\mu} = 0.
\end{align*}
This occurs exactly when $P_i N A \dot{\mu}$ and $q_i$ are parallel, or rather,
\begin{align*}
    P_i N A \dot{\mu} = \alpha P_i(N A \mu + G^\dag w),
\end{align*}
for some scalar $\alpha \in \nR{}$.
To preclude the possibility that $P_i N A \dot{\mu} = 0$, it is necessary and sufficient that $P_i N A$ is rank 2, which is almost certain for a randomly chosen $A$.
Then, the equation admits a solution for $\alpha \neq 0$ only if
\begin{align*}
    P_i G^\dag w = P_i N A(\alpha^{-1}\dot{\mu} + \mu) \in \image{P_i N A}.
\end{align*}
This is a 2D subspace of $\nR{3}$, and thus the probability that $P_i G^\dag w$ lies in $\image{P_i N A}$ is zero.
The only exception is when $P_i G^\dag w = 0$, but in this case there is still no solution since $ P_i N A(\alpha^{-1}\dot{\mu} + \mu) \neq 0$ for all time.
This shows that a randomly chosen matrix $A$ will almost certainly satisfy the second condition.

For the first condition, $\lambda$ fails to belong to $\mathcal{M}_\lambda$ if
\begin{align*}
    P_i G^\dag w + P_i N A \mu = 0,
\end{align*}
which similarly requires that $P_i G^\dag w \in \image{P_i N A}$.
This also has probability zero for a randomly chosen matrix $A$.
We conclude that if $\lambda(t)$ is determined by \eqref{eq:linear_solution} for a randomly chosen $A$, then it solves Problem~\ref{problem} almost surely.

\section{Simulations}

\begin{figure*}[t]
    \centering
    \begin{subfigure}[]{0.9\textwidth}      \centering
        \includegraphics[width=0.45\textwidth]{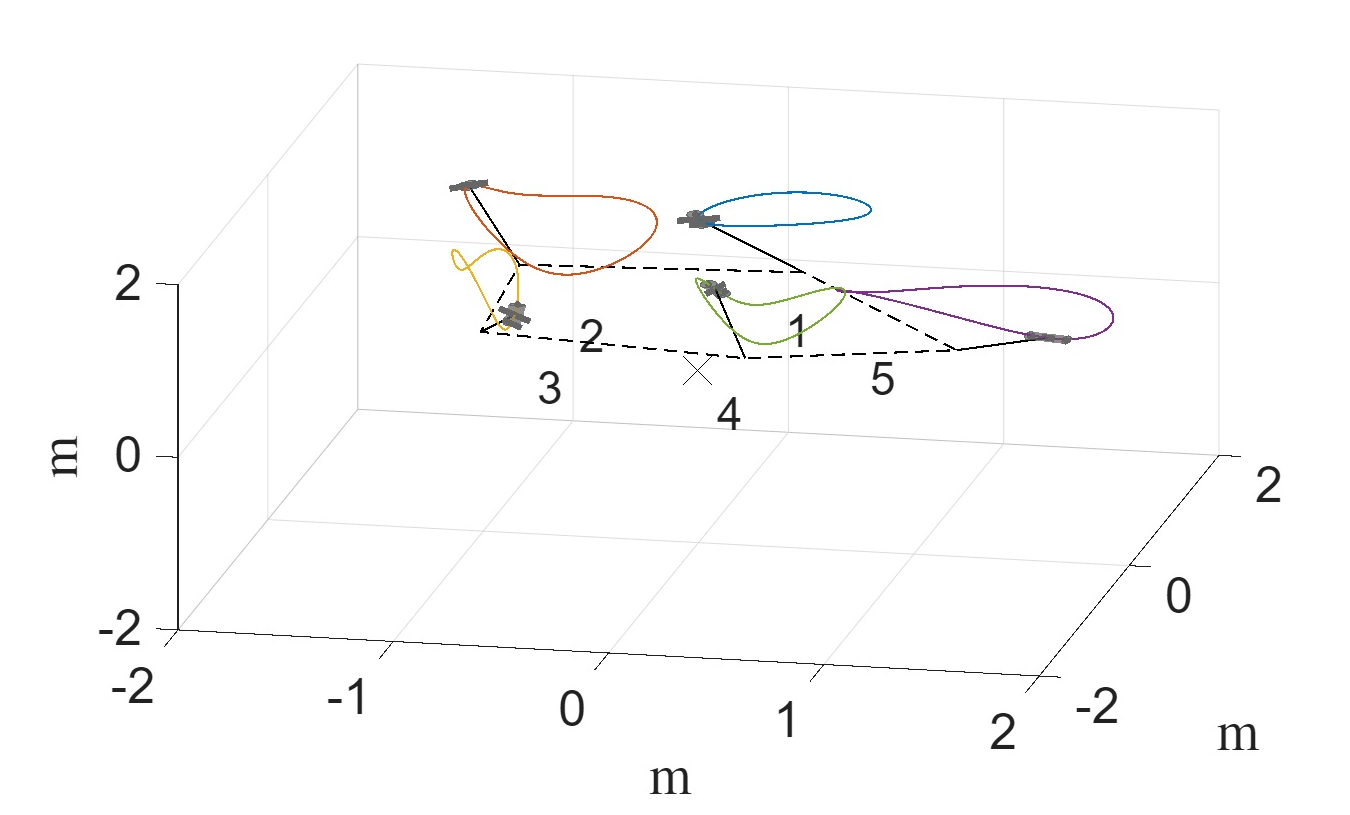}\qquad
        \includegraphics[width=0.40\textwidth]{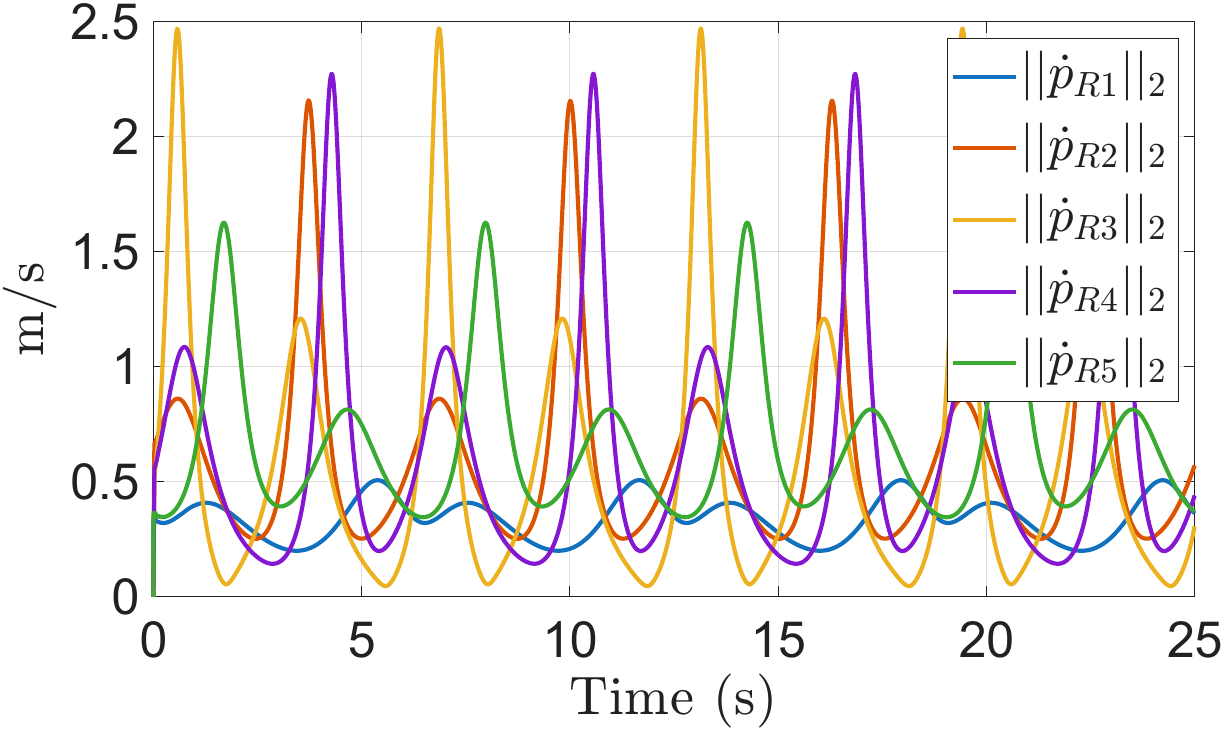}\\
         \includegraphics[width=0.45\textwidth]{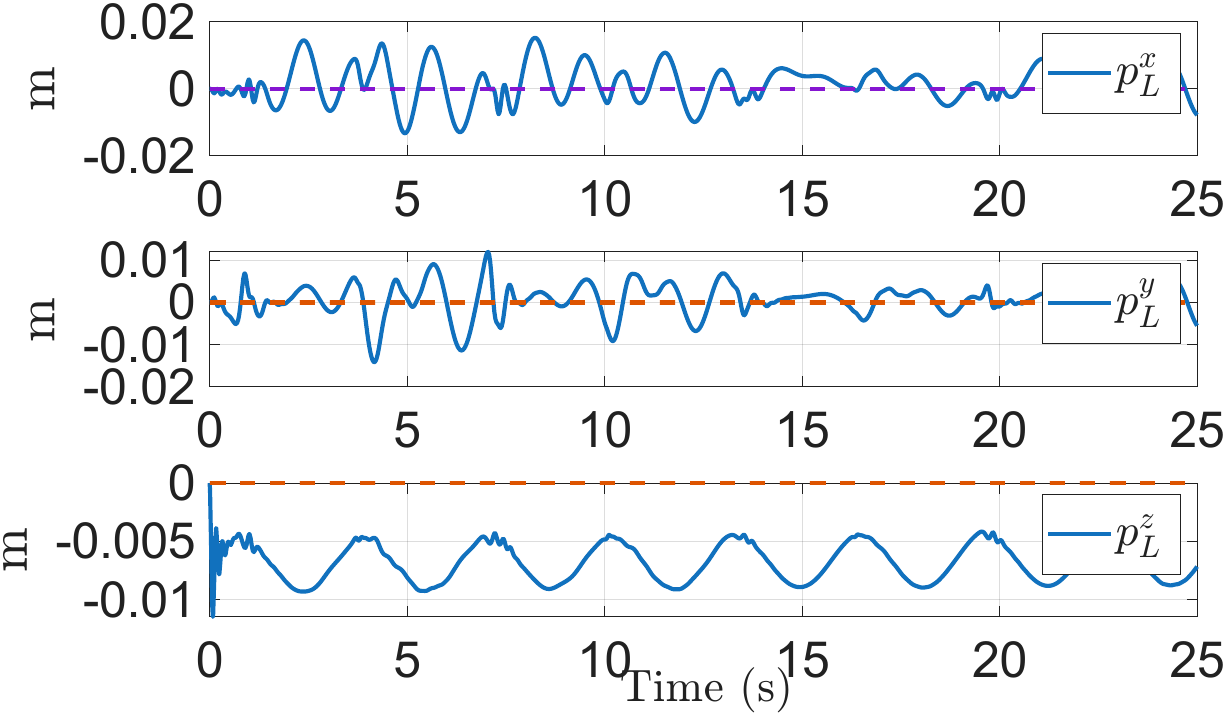}
          \qquad\includegraphics[width=0.45\textwidth]{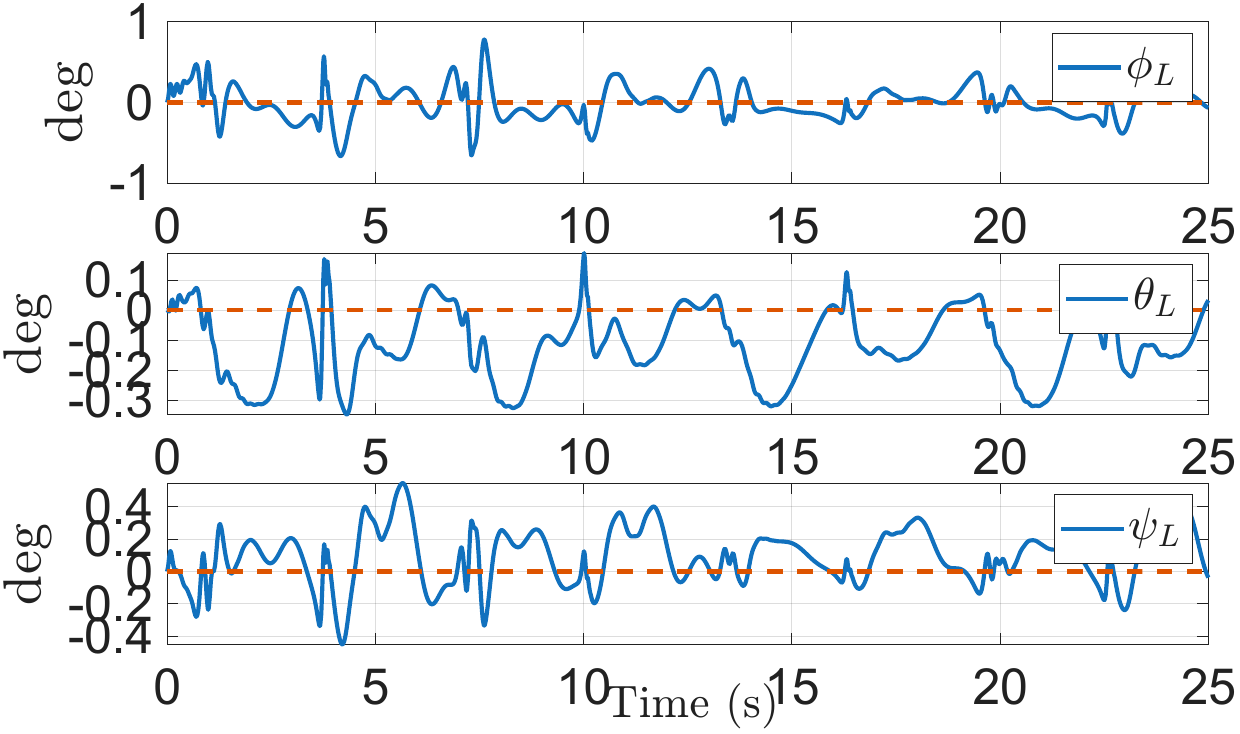}
        \caption{$n=5$\label{fig:n5}}
        \label{fig:group1}
    \end{subfigure}
    \hfill
    \begin{subfigure}[]{\textwidth}  \centering   \includegraphics[width=0.45\textwidth]{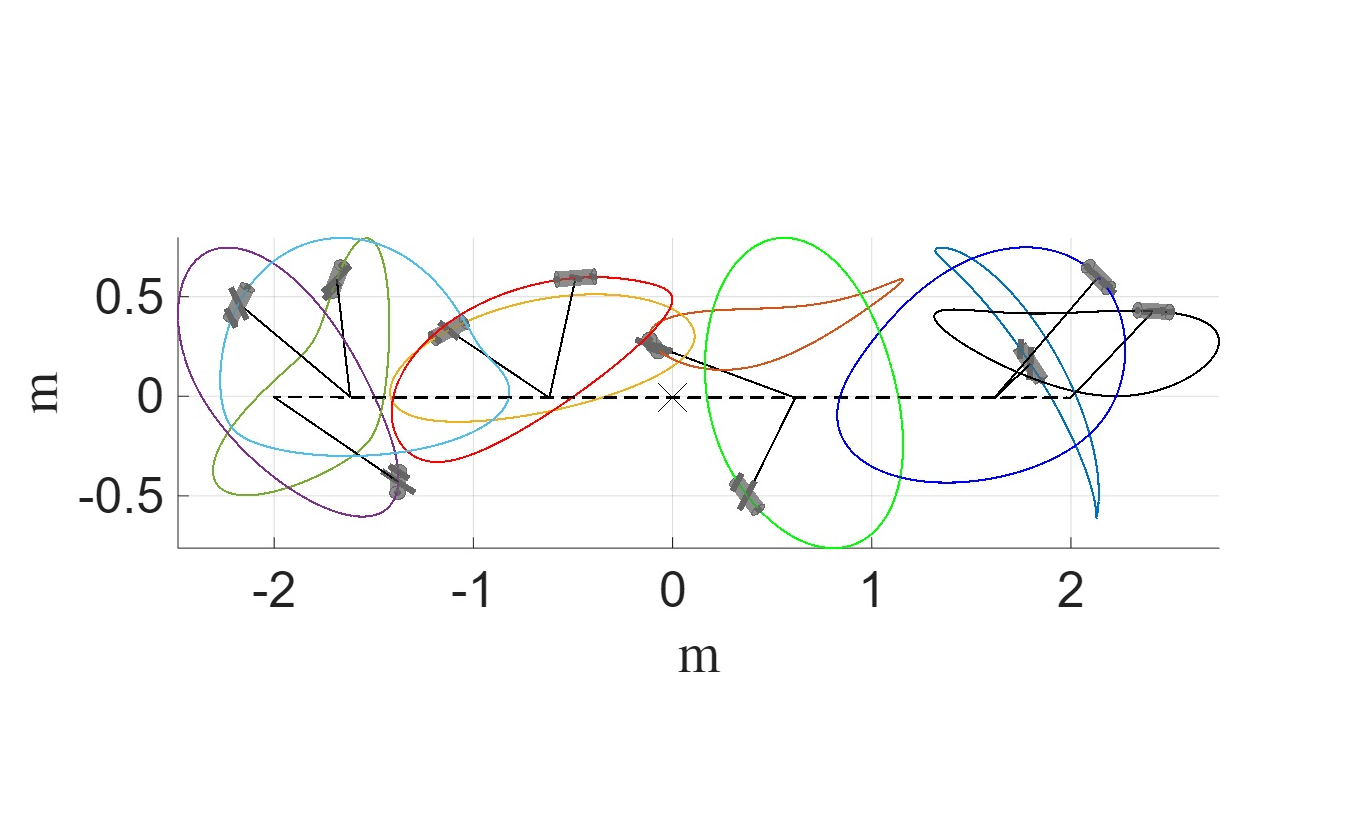}\qquad
        \includegraphics[width=0.40\textwidth]{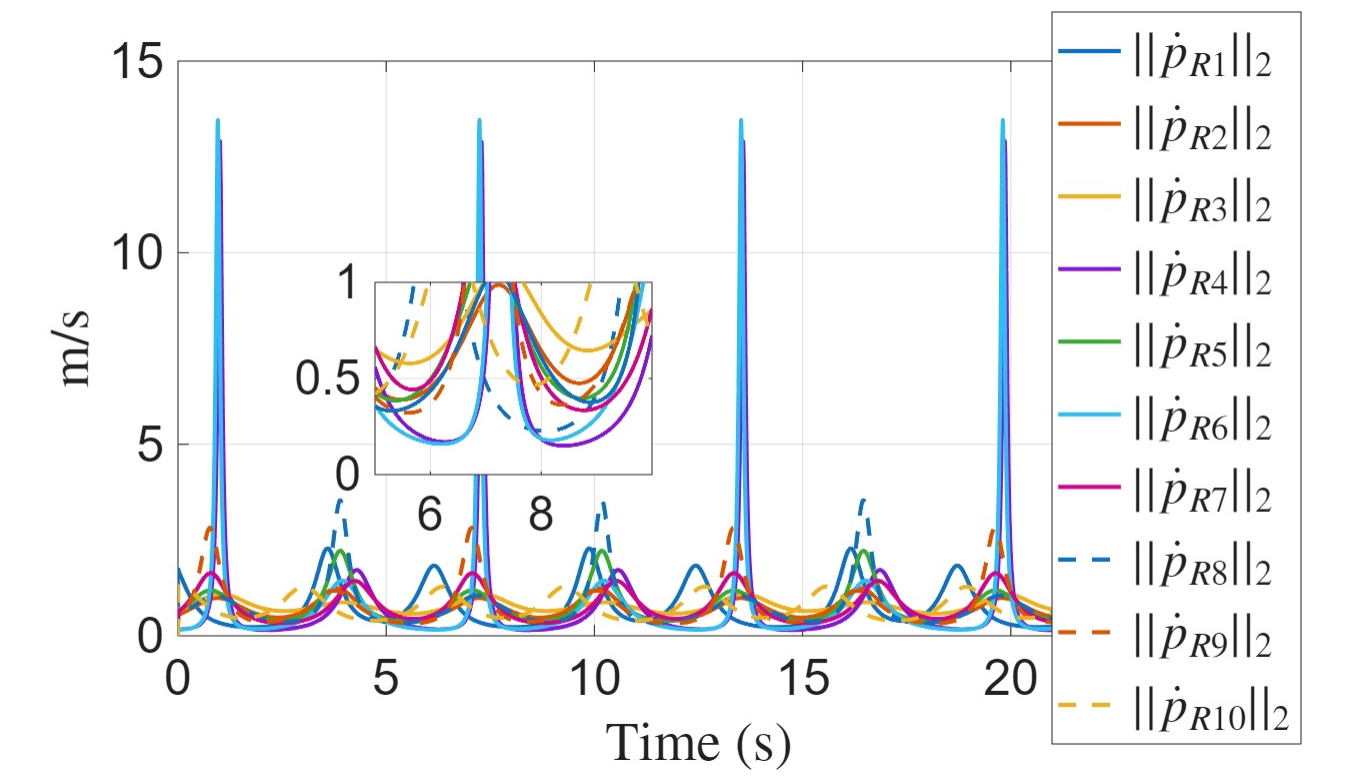}\\
         \includegraphics[width=0.45\textwidth]{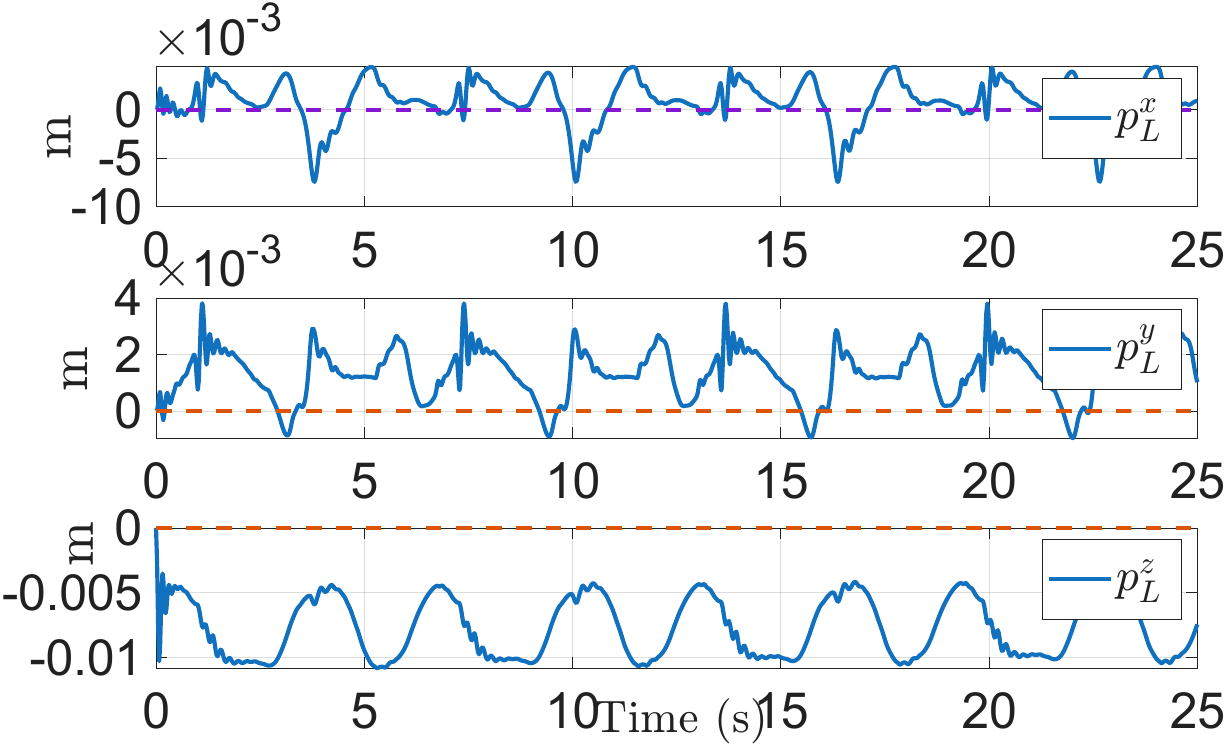}
         \qquad \includegraphics[width=0.45\textwidth]{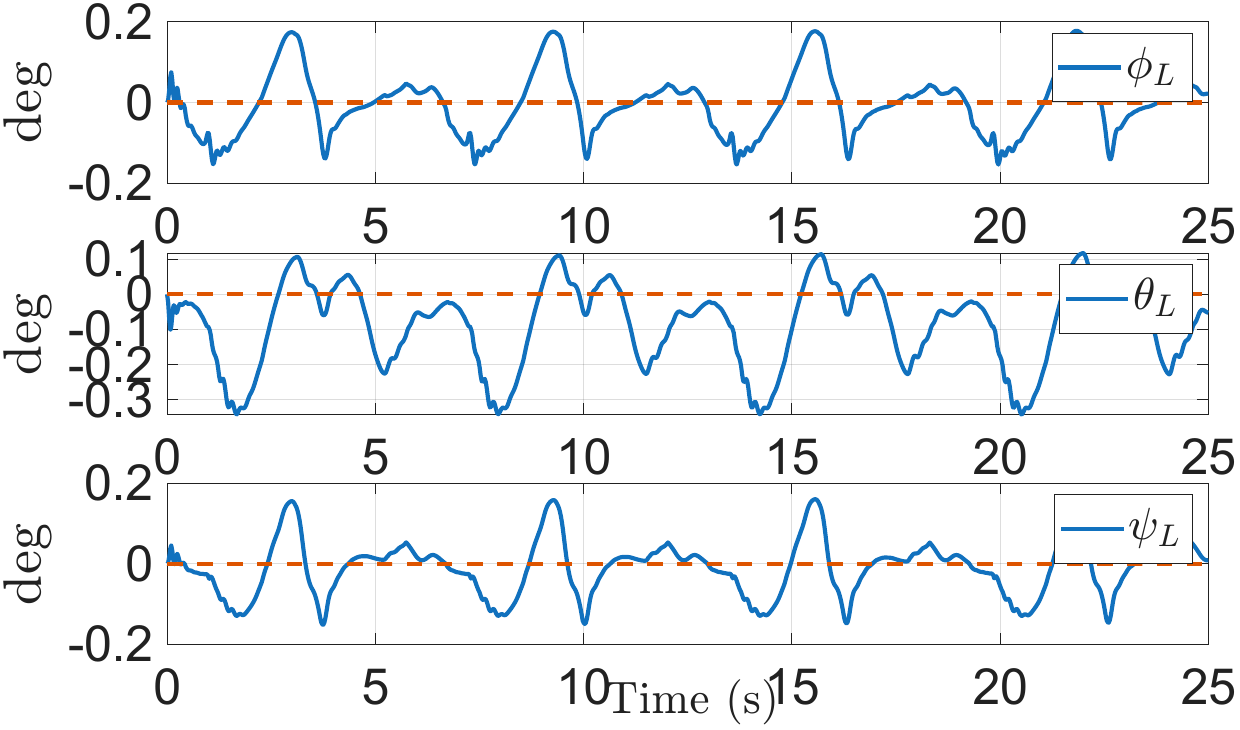}
        \caption{$n=10$, with cables attached uniformly around the center of mass of the load. At the top left: XZ side view of the system. \label{fig:n10}}
    \end{subfigure}
    
    \caption{Results of two simulations. A representation of the system shows the colored orbits of the carriers, on which the carriers are represented in grey; black solid lines are the cables and black dashed lines connect the cable attachment points on the load, which center of mass is a black cross. $p_L^{X,Y,Z}$ are the X-, Y-, and Z- coordinated of the position of $\{O\}$ expressed in $\{W\}$, while $\phi_L,\theta_L,\psi_L$ the roll, pitch, and yaw angles extracted from $R$.}
    \label{fig:main}
\end{figure*}

In this section, we demonstrate an application of the proposed theory to generate trajectories that were not possible to achieve via prior methods.
The simulation implements realistic position control for the carriers and dynamics for the manipulated load to demonstrate the applicability of the proposed methodology.
The code used to generate the simulation results is available at {\footnotesize \url{https://codeberg.org/CGabelli/Coordinated-Trajectories-Nonstop-Carriers-Cable-Suspended-Load}}.

\subsection{Setup}
The trajectory planning method proposed in \eqref{eq:linear_solution} has been tested in a MATLAB/Simulink environment. A rigid-body load of mass $0.5 \unit{\kilogram}$ and rotational inertia $0.01I_3\unit{\kilogram\cdot\meter^2}$ is manipulated by $n$ cables attached to aerial carriers. The carriers are modelled as gravity-compensated second-order controlled systems subject to the cable force acting on the carrier, $-f$. The desired positions and velocities of the carriers are indicated as $p_i^d$, $\dot{p}_i^d$, and are computed with the proposed method from the system kinematics in \eqref{eq:kinematics}. Specifically, the $i$-th carrier's dynamics is given by
$$ \ddot{{p}}_i = M_i^{-1} \left(-{f}_i + K_{1}  \dot{{e}}_i + K_2  {e}_i +K_3 \int_0^t{{e}_i d\tau} \right),$$
with ${e}_i=({p}_i ^d - {p}_i)$, $M_i=0.01 I_3\unit{\kilogram}$ 
the inertia of the carrier, and $K_1=500 I_3\unit{\kilogram\per\second^2}$, $K_2=10 I_3\unit{\kilogram\per\second}$, and $K_3=20 I_3\unit{\kilogram\per\second^3}$ the control gains. Note that the cable force is assumed uncompensated, and so it is treated as an external disturbance rejected by the trajectory-tracking controller. The tracking could be further improved by compensating the cable force, if one assumed it is measured or estimated, as done, e.g., in \cite{gabellieri2023equilibria}.
To model non-idealities of real cables and for the modularity of the simulator, each cable is represented by a spring-damper system, while the trajectory generation algorithms ignores the cable elongation. The cable rigidity is ${K_c=700\unit{\kilogram\per\second^2}}$, the damping coefficient is ${B_c=1\unit{\kilogram\per\second}}$, and rest length is ${l_0=0.8\unit{\meter}}$. Define an inertial frame $\{W\}$, the load pose is initialized at zero position and attitude $R=I_3$, where $R\in SO(3)$ expresses the orientation of $\{O\}$ w.r.t. $\{W\}$, and it should be kept constant during the simulation. 

\subsection{Results}
Figures \ref{fig:main} show the results of the linear solution from Sec. \ref{sec:lin_sol}, where matrix $A$ has been generated picking one out of 1000 random matrices of entries in $[0,1]$ with the lowest condition number. In particular, Figure \ref{fig:n5} reports results for a system with $n=5$ cables. Each cable's attachment point on the load is chosen such that its position is ${b_i=[R_z(\frac{2\pi}{n}i+0.2r)[1.2,0]^\top, r']^\top}$, where $r$, is a random angle in $[0, -0.2]\unit{\radian}$ and $r'$ a random length in $[0,1]\unit{\meter}$.
Figure \ref{fig:n10} shows the results obtained for $r,r'$ equal to zero, namely, for cables attached on a circle of radius $1.2\unit{\meter}$ around the load center of mass. Such an example is relevant for practical applications, where it may be desirable to attach the cables uniformly around the center of mass of an extended rigid body and to maintain it horizontal.

\subsection{Discussion}
In all simulation results, the load's pose remains steady during the manipulation, while the carriers keep travelling with non-zero velocity.
Small oscillations in the load pose are caused by the simulated carrier dynamics, resulting in slight deviations of their trajectories from the planned ones.
Furthermore, a small error in the vertical positioning of the load is due to the uncompensated and varying cable elongation, of which the algorithm is unaware, as it is realistic to assume in practical applications. 
The error depends on the cable stiffness, which is likely very high for real cables, leading to negligible elongations. 
Such an error in the load pose is expected, since the proposed method does not close the loop on the object's pose but generates offline trajectories that are then tracked in the simulator by a controller that closes the loop on the carriers' state.
In the future, the proposed trajectory generation method may be integrated in an online load pose controller similarly to \cite{girardello2025trajectory} if one wishes to exploit it to regulate the load pose along a time-dependent trajectory.

The example shown in Figure \ref{fig:n10} shows that the carriers' paths cross the plane on which the attachment points on the load lie, which would not be possible using previous methods \cite{gabellieri2025coordinated}, where the planned carrier trajectories were constrained to one half-space by construction. 
The results show that the proposed method successfully plans for the first time general non-stopping trajectories for the carriers capable of keeping the load pose unperturbed. The increased workspace of the carriers opens new possibilities for future research; e.g., it could be exploited for planning in cluttered environments or for modulating the amount of the external wrench $w$ compensated by some of the carriers, in response to energy saving requirements.

\addtolength{\textheight}{-5cm}   

\section{Conclusions}

This work explored the problem of using multiple carriers with nonzero velocity and smooth periodic trajectories to manipulate a load through cable attachments.
We showed that the problem of finding suitable trajectories could be translated to finding an immersion from the unit circle $S^1$ into the manifold of admissible configurations $\mathcal{M}_\lambda$.
Moreover, we showed that $\mathcal{M}_\lambda$ is connected, meaning any two admissible configurations can be connected by a smooth trajectory, and that the map $q_i : \mathcal{M}_\lambda \to S^2$ from the admissible configuration to a given individual carrier is full rank.
Finally, we presented a simple linear solution based on our theory and simulation results that show it is able to achieve trajectories not possible using prior methods.
This work provides significant new insights into the multiple carrier problem that can be used to support novel control strategies in future work.


\bibliographystyle{IEEEtran}
\bibliography{biblio}

\section*{Appendix}

\begin{lemma}\label{lem:connected_submanifold}
    Let $\mathcal{M}$ be a connected manifold of dimension $n$ and let $\mathcal{S} \subset \mathcal{M}$ be an embedded submanifold of dimension at most $n-2$.
    Then $\mathcal{M} \setminus \mathcal{S}$ is a submanifold of $\mathcal{M}$ with dimension $n$, and it is (path-)connected.
\end{lemma}

\begin{proof}
    The fact that $\mathcal{M} \setminus \mathcal{S}$ is a submanifold of dimension $n$ is immediate since it is an open subset of $\mathcal{M}$ with an inherited topology.

    We will assume that $\mathcal{M} \setminus \mathcal{S}$ is disconnected and show that this leads to a contradiction.
    If $\mathcal{M} \setminus \mathcal{S}$ is disconnected then it is equal to $A_1 \sqcup A_2$ where $A_1, A_2 \subset \mathcal{M} \setminus \mathcal{S}$ are disjoint open sets.
    These sets are also open in $\mathcal{M}$ by inclusion, and we have that $\mathcal{M} = A_1 \sqcup A_2 \sqcup S$.
    Let $\varphi : U \to \mathbb{R}^n$ be a slice chart for $\mathcal{S}$ in $\mathcal{M}$, i.e.
    \begin{align*}
        &\varphi(U \cap \mathcal{S}) = \\ &\cset{(x^1,...,x^s, x^{s+1}, ..., s^n) \in \varphi(U)}{x^{s+1} = \cdots = x^n = 0},
    \end{align*}
    where $s = \dim \mathcal{M} - \dim \mathcal{S}$. 
    Suppose additionally and $U \cap A_1$ and $U \cap A_2$ are nonempty, noting that such a slice chart must exist, since if it did not, then $\mathcal{M}$ would be disconnected.
    Thus $U \setminus S = (U \cap A_1) \sqcup (U \cap A_2)$ is disconnected as well.
    Without loss of generality, we may assume that the image of $\varphi(U)$ is the open ball of radius 1.
    
    Let $x_0, x_1 \in \varphi(U \setminus \mathcal{S}) = B_1(0)$ be arbitrary.
    We will show that $x_0$ and $x_1$ are connected.
    The straight line connecting $x_0$ and $x_1$ is given by $x(t) = x_0 + t (x_1 - x_0)$.
    Let $P = \begin{pmatrix} 0_{n-s \times s} & I_{n-s} \end{pmatrix} \in \nR{n-s \times n}$ be the projection onto the last $n-s$ coordinates.
    Then $x(t)$ intersects $\varphi(U \cap \mathcal{S})$ if and only if $P (x_0 + t (x_1 - x_0)) = 0$ for some $t \in [0,1]$.
    This requires that $P x_0$ and $P (x_1-x_0)$ are parallel.
    If this is not the case, then $x_0$ and $x_1$ are connected and we are done.
    Otherwise, choose any direction $z$ so that $P z$ is not parallel to $P x_0$, and let $\hat{x}_1 = x_1 + \varepsilon z$ with $\varepsilon > 0$ sufficiently small so that $\hat{x}_1 \in \varphi(U \setminus \mathcal{S})$ and $x_1$ is connected to $\hat{x}_1$ by a straight line.
    Then $P x_0$ is not parallel to $P(\hat{x}_1 - x_0)$ by construction, so $x_0$ is connected to $\hat{x}_1$ and hence to $x_1$.
    
    We have shown that $\varphi(U \setminus \mathcal{S})$ is connected, and therefore $U \setminus \mathcal{S}$ is connected.
    But this is in contradiction to our decomposition of $U \setminus S = (U \cap A_1) \sqcup (U \cap A_2)$.
    Therefore the original manifold $\mathcal{M} \setminus \mathcal{S}$ cannot be disconnected, and hence must be connected.
\end{proof}

\end{document}